\documentclass{article}
\usepackage{algorithm}
\usepackage{algpseudocode}
\usepackage{graphicx}
\usepackage{amssymb}
\usepackage{amsthm}

\newtheorem{dfn}{Definition}
\newtheorem{thm}{Theorem}

\newcommand{\ie}{i.e.\ }
\newcommand{\depend}[1]{\mathit{depend}(#1)}
\newcommand{\source}[1]{\mathit{source}(#1)}
\newcommand{\target}[1]{\mathit{target}(#1)}

\begin{document}

\title{Probabilistic Analysis of Loss in Interface Adapter Chaining}

\author{Yoo Chung \and Dongman Lee}

\maketitle

\begin{abstract}
  Interface adapters allow applications written for one interface to be reused with another interface without having to rewrite application code, and chaining interface adapters can significantly reduce the development effort required to create the adapters.  However, interface adapters will often be unable to convert interfaces perfectly, so there must be a way to analyze the loss from interface adapter chains in order to improve the quality of interface adaptation.  This paper describes a probabilistic approach to analyzing loss in interface adapter chains, which not only models whether a method can be adapted but also how well methods can be adapted.  We also show that probabilistic optimal adapter chaining is an NP-complete problem, so we describe a greedy algorithm which can construct an optimal interface adapter chain with exponential time in the worst case.
\end{abstract}

\section{Introduction}
\label{sec:introduction}

Network services are being developed all the time, along with the interfaces that specify how these services should be accessed.  Only a very small number of the interfaces to these services are standardized, and many interfaces can be developed for different services which have very similar functionality.  In order to access a different interface than a client was written for without rewriting the client, interface adapters could be used to convert invocations in one interface to another, which can also be chained to reduce the number of interface adapters that must be created~\cite{h:gamma1994,vayssiere:doa2001,gschwind:sem2002,motahari-nezhad:www2007,kim:percom2008}.

However, it is unlikely that interface adaptation can be done perfectly, since interfaces are usually developed independently of each other with no regard for compatibility.  Adaptation loss will usually result as certain methods cannot be adapted by the interface adapter, and the problem is only worse when adapters are chained.  Even analyzing how much loss results from an interface adapter chain is not a trivial problem that can be modeled as a shortest path problem.

Our previous work~\cite{kim:percom2008,chung:ietsoftw2010} took the approach of assuming that a method in a target interface could be implemented as long as all the prerequisite methods in the source interface were available.  However, a discrete approach such as this ignores the possibility of \emph{partial} adaptation of methods, where an adapted method may not be able to be invoked with all possible arguments because of limitations with methods in a source interface.  For a trivial example, negative numbers for a square root function cannot be handled if either the source interface or target interface are unaware of imaginary numbers.

We describe a probabilistic approach to handling the partial adaptation of methods, where the loss may occur not just due to missing functionality or methods, but also due to an interface adapter being unable to handle all arguments given for a method in a target interface.  We first investigate how probabilities should be expressed, where independence assumptions are made so that we can obtain a computational model that can be feasibly used in a real system.  Based on this probabilistic model, we define how to express probabilistic loss in interface adaptation and how to model interface adapters, which can then be used to probabilistically analyze loss in interface adapter chains.  As in the discrete approach~\cite{chung:ietsoftw2010}, probabilistic optimal adapter chaining is NP-complete, so we describe a greedy algorithm which can construct an optimal adapter chain with exponential run-time in the worst case.

This paper is structured as follows.  In section~\ref{sec:preliminaries}, we describe elements from the discrete approach which we use in developing the probabilistic approach.   In section~\ref{sec:probabilistic-math}, we formulate the probabilistic approach for analyzing loss in interface adapter chains.  Section~\ref{sec:complexity-probabilistic} shows that probabilistic optimal adapter chaining is NP-complete, and section~\ref{sec:greedy-probabilistic-algorithm} describes an algorithm which can construct an optimal adapter chain with exponential run-time in the worst case.  We discuss related work in section~\ref{sec:related-work}, and section~\ref{sec:conclusions} concludes.

\section{Related work}
\label{sec:related-work}

Ponnekanti and Fox~\cite{ponnekanti:percom2003} suggests using interface adapter chaining for network services to handle the different interfaces available for similar types of services.  They provide a way to query all services whose interfaces can be adapted to a known interface.  They also support lossy adapters, but the support is limited to detecting whether a particular method and specific parameters can be handled at runtime.  They do not provide a way to analyze the loss of an interface adapter chain, so they are unable to choose a chain with less loss when alternatives are available.

Gschwind~\cite{gschwind:sem2002} allows components to be accessed through a foreign interface and implements an interface adaptation system for Enterprise JavaBeans~\cite{ejb}.  It implements a centralized adapter repository that stores adapters, along with weights that mark the priority of an adapter.  Dijkstra's algorithm~\cite{dijkstra:mathematik1959} is used to construct the shortest interface adapter chain that adapts a source interface into a target interface.  While there is support for marking an adapter as lossy or not, it does not have the capability to properly analyze and compare the loss in interface adapter chains.

Vayssi\'{e}re~\cite{vayssiere:doa2001}  supports the interface adaptation of proxy objects for Jini~\cite{jini}.  The goal is to enable clients to use services even when they have different interfaces than expected.  It provides an adapter service which hooks into the lookup service, so that a client can use a proxy object without having to be aware that any adaptation occurs.  No consideration is spent on the possibility that interface adapters may not be perfect.

There is also other work using chained interface adapters which focus on maintaining backward compatibility as interfaces evolve~\cite{keller:ecoop1998,hallberg:monadreader2005,kaminski:cascon2006}.  Since these are applied to different versions of the same interface, they do not consider the possibility of adaptation loss, in contrast to other work where the focus is on adaptation between different interfaces with potentially irreconcilable incompatibilities.

\section{Preliminaries}
\label{sec:preliminaries}

In this section, we describe the bare essentials from a discrete approach of analyzing lossy interface adapter chaining~\cite{chung:ietsoftw2010}, which are necessary for the probabilistic approach developed in section~\ref{sec:probabilistic-math}.   In this section as well as in the rest of the paper, a range convention for the index notation used to express matrixes and vectors will also be in effect~\cite{index-notation}.

We take the view that an interface is a specification of a collection of methods (which can also be called operations, methods, member functions, etc.)\@ which specify the concrete syntax and types for invoking actions on a service (which can also be called an object, module, component, etc.)\@ that conforms to the interface.

An interface adapter transforms calls for one interface into calls for another.  For example, if one interface has a method \texttt{set\-Audio\-Properties} while another interface has methods \texttt{set\-Volume} and \texttt{set\-Balance}, an interface adapter could handle a call to \texttt{set\-Audio\-Property} with the former interface using calls to \texttt{set\-Volume} and \texttt{set\-Balance} when the actual service conforms to the latter interface.

Adapting interfaces using a chain of interface adapters means converting calls to an interface to another using one interface adapter, then converting them again to yet another interface with a subsequent interface adapter, and so on until we can convert calls for a desired source interface to calls for a desired target interface.

A method dependency matrix is used to express the methods in a source interface necessary for providing methods in a target interface:
\begin{dfn}
  \label{dfn:dependency-matrix}
  A \emph{method dependency matrix} $a_{ji}$ is a boolean matrix where:
  \begin{itemize}
  \item $a_{11}$ is true, while $a_{1i}$ is set to false for all $i \neq 1$.
  \item If method~$j$ can always be implemented in the target interface, set $a_{ji}$ to false for all $i$.
  \item If method~$j$ can never be implemented given the source interface, set $a_{j1}$ to true, while $a_{ji}$ is set to false for all $i \neq 1$.
  \item If method~$j$ depends on the availability of actual methods in the source interface, then $a_{j1}$ is false, while $a_{ji}$ is true if and only if method~$j$ in the target interface can be implemented only if method~$i$ in the source interface is available.
  \end{itemize}
\end{dfn}

Method dependency matrixes can be composed, which in effect models two interface adapters chained together as a single equivalent adapter in terms of loss:
\begin{dfn}
  Given method dependency matrixes $b_{kj}$ and $a_{ji}$, the \emph{composition operator} $\otimes$ of two method dependency matrixes is defined as:
  \begin{equation}
    b_{kj} \otimes a_{ji} = \bigvee_j (b_{kj} \wedge a_{ji})
  \end{equation}
\end{dfn}

\begin{thm}
The composition operator for method dependency matrixes is associative:
\[ c_{lk} \otimes (b_{kj} \otimes a_{ji}) = (c_{lk} \otimes b_{kj}) \otimes a_{ji} \]
\end{thm}

We can also define an interface adapter graph, which is a directed graph where interfaces are nodes and adapters are edges.  If there are interfaces $I_1$ and $I_2$ with an adapter $A$ that adapts source interface $I_1$ to target interface $I_2$, then $I_1$ and $I_2$ would be nodes in the interface adapter graph while $A$ would be a directed edge from $I_1$ to $I_2$.

\begin{dfn}
  An \emph{interface adapter graph} is a directed graph where interfaces are nodes and adapters are edges.  The source node for an edge corresponds to the source interface, while the target node for an edge corresponds to the target interface.
\end{dfn}

\section{Probabilistic analysis}
\label{sec:probabilistic-math}

When an interface adapter translates a call to a method in the target interface to calls to a method in the source interface, it is possible that the translation cannot be done perfectly.  If the source interface lacks certain capabilities, then the adapter may not be able to properly process specific parameters received by a method.

\begin{figure}
  \centering
  \includegraphics[width=10cm]{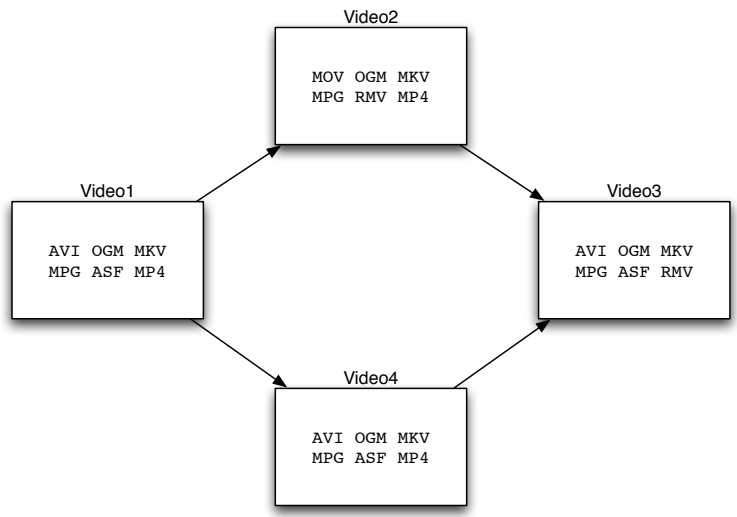}
  \caption{Adapting playback methods in video playback interfaces which support different video formats.}
  \label{fig:example}
\end{figure}

For example, there may be multiple video playback interfaces with adapters between them as in figure~\ref{fig:example}, where each interface is only able to handle a specific set of video formats.  For instance, if client code is written to access interface \textit{Video2} but the actual service has interface \textit{Video1}, then parameters to the playback method for \textit{Video2} with formats \texttt{MOV} and \texttt{RMV} cannot be handled properly.\footnote{We ignore the possibility that video conversion could be done by the adapter itself.}  In another situation where client code is written for \textit{Video3} but the actual service conforms to \textit{Video1}, we would need an interface adapter chain from \textit{Video1} to \textit{Video3}, and we would like to know if the chain that goes through \textit{Video2} or the one that goes through \textit{Video4} is better.

With the discrete approach, which only looks at whether methods are available or not, we must make a choice about what to do with methods that can be only partially adapted.  Conservatively treating such methods as being unavailable excludes the use of interface adapter chains that can do an imperfect but mostly complete job of adapting such methods.  On the other hand, optimistically treating such methods as being available could result in the selection of an interface adapter chain that is much worse than other chains in terms of how complete the adaptation is.  In figure~\ref{fig:example}, \textit{Video3} could not be adapted from \textit{Video1} at all with the conservative treatment, while with the optimistic treatment we would not be able to determine that the chain which goes through \textit{Video4}, which can support \texttt{AVI}, \texttt{OGM}, \texttt{MKV}, \texttt{MPG}, and \texttt{ASF}, is superior than the one that goes through \textit{Video2}, which can only support \texttt{OGM}, \texttt{MKV}, and \texttt{MPG}.

The probabilistic approach we introduce here takes into account that methods can be partially adapted, relaxing the binary limitation of only treating a method as available or not.

\subsection{Probabilistically modeling interface adaptation}
\label{sec:prob-model}

We develop a probabilistic approach by starting off with the most general form of expressing the probabilities and adding assumptions until we have a probabilistic formula that is practical.  Without additional assumptions, the probabilities can only be expressed in a way that is useless for analyzing real systems.  The additional assumptions allow us to express the desired probabilities in a way that they can be feasibly computed from a set of values that can be measured in practice.\footnote{There is a more precise approach using abstract interpretation that does not rely on such assumptions, but it is much more difficult to set up and requires exponential space complexity~\cite{l:chung2010:absintp}.}

We first describe the notation for expressing certain probabilistic events in table~\ref{tab:probabilistic-events}.  These events denote whether a method can handle a given argument, or whether an interface adapter can convert an argument for a method in a target interface to an argument for a method in the source interface and successfully convert back the result.  We assume that a method only accepts a single argument: this is not a problem since methods with multiple arguments can simply be modeled as a method accepting a single tuple with multiple components.  If a method does not need an argument, we treat it as receiving a dummy argument anyway.

\begin{table}
  \centering
  \begin{tabular}{cp{9cm}}
    $V_{m,I}(a)$ & Method~$m$ of interface~$I$ can properly handle argument~$a$. \\ \\
    $V_{m,I}$ & Method~$m$ of interface~$I$ can properly handle its argument. \\ \\
    $C^A_{m \rightarrow m'}$ & Interface adapter~$A$ can successfully convert an argument for method~$m$ in the target interface to an argument for method~$m'$ in the source interface and convert back the result.
  \end{tabular}
  \caption{Probabilistic events.}
  \label{tab:probabilistic-events}
\end{table}

Let us say that we wish to adapt methods in source interface~$I_S$ into method~$j$ in target interface~$I_T$.  The most general form for expressing the probability that a method could handle an argument is to sum the probabilities for every possible argument, where we must consider the probability of the method receiving a specific argument and then the probability that the method can handle it:

\begin{equation}
  \label{eq:general-probability}
  P(V_{j,I_T}) = \sum_a P(V_{j,I_T}(a)) P(A = a)
\end{equation}

The most general form for expressing the probability requires that we know the probability distribution of arguments, which is not feasible except for the simplest of argument domains.  For example, the probability distribution for a simple integer argument may require $2^{32}$ or $2^{64}$ probabilities to be expressed for the typical computer architecture, and even measuring such a probability distribution may not be feasible in the first place.  It is also not feasible that we already know the probabilities for how a method can handle each and every possible argument.

For this reason, we make the assumption that the probabilities do \emph{not} depend on the specific arguments.  Given this assumption, we can now express $P(V_{j,I_T})$ in terms of whether an argument can be converted and whether it can be handled.  More specifically, this means that for \emph{all} methods in the source interface that the interface adapter~$A$ requires to implement a method in the target interface, it must be the case that the argument can be converted \emph{and} the method in the source interface can handle the converted argument.  Using the method dependency matrix $a_{ji}$ for adapter~$A$, $P(V_{j,I_T})$ can be expressed as:

\begin{equation}
  \label{eq:prob-ind-args}
  P(V_{j,I_T}) = P\left( \bigcap_{a_{ji}} \left(V_{i,I_S} \cap C^A_{j \rightarrow i}\right) \right)
\end{equation}

This is still too unwieldy an expression to be practical, since it is unclear how dependencies in the events for different methods in the source interface affect the overall probability.  It would also be unclear how to measure the probabilities beforehand without trying out every possible argument and configuration of interface adapter chains, something that is clearly not feasible. Therefore we make an additional assumption that the events for separate methods in the source interface are independent.

With the additional assumption, $P(V_{j,I_T})$ can be expressed as:

\begin{equation}
  \label{eq:prob-indy-meth}
  P(V_{j,I_T}) = \prod_{a_{ji}} P( V_{i,I_S} \cap C^A_{j \rightarrow i} )
\end{equation}

However, equation~(\ref{eq:prob-indy-meth}) is still not appropriate for practical use.  The reason is that it entangles the work done by the interface adapter and whether the method in the source interface can handle the converted argument.  Basically, the probabilities intrinsic to the interface adapter and the source interface are entangled.  If the source interface itself is the result of adaptation through an interface adapter chain, then we have the problem of a configuration-dependent event being entangled with a configuration-independent event, and there is no simple way to derive the required probabilities.

Thus we make one final additional assumption that the probability an interface adapter can successfully convert arguments and results is independent from the probability that a method in the source interface can handle an argument.  This allows us to express $P(V_{j,I_T})$ as:

\begin{equation}
  \label{eq:practical-probability}
  P(V_{j,I_T}) = \prod_{a_{ji}} P( V_{i,I_S} ) P( C^A_{j \rightarrow i} )
\end{equation}

Equation~(\ref{eq:practical-probability}) is finally in a form that can be used practically.  The probability that an interface adapter~$A$ can successfully convert an argument for method~$j$ in the target interface to an argument for method~$i$ in the source interface, $P( C^A_{j \rightarrow i} )$, is a value that is intrinsic to an interface adapter.  In principle, it could be measured empirically by exhaustively testing the interface adapter to see which arguments it can accept, although in practice more sophisticated testing based on random samples would be used.  It might even be possible to obtain the probabilities through analysis of the interface adapter code.  The probability that method~$m_i$ in source interface~$I_S$ can handle an argument, $P( V_{i,I_S} )$, is also a value that can be obtained, either through analytical or empirical means similar to measuring probabilities from interface adapters if $I_S$ is an interface to an actual service, or through a recursive application of equation~(\ref{eq:practical-probability}) when $I_S$ is an adapted interface.

\subsection{Formalizing adapter loss}
\label{sec:prob-formalism}

We now have the basis for describing a framework similar to the one developed for the discrete chain approach.  We define a method availability vector and a method dependency matrix, but in addition we also define a \emph{conversion probability matrix}.

As before, the method availability vector~$p_i$ expresses how well a method is supported in an interface, and it is not intrinsic to an interface but rather represents the loss from interface adaptation.  The components for a method availability vector in the probabilistic approach are probabilities.  $p_i$ is defined as the probability that method~$i$ can handle an argument it receives, \ie $p_i = P(V_{i,I})$.

The method dependency matrix is the same as defined in section~\ref{sec:preliminaries} and is used in equation~(\ref{eq:practical-probability}).  Unlike for the discrete chain approach, however, the method dependency matrix does not suffice to describe the relevant information for an interface adapter.  We also require a set of probabilities~$P(C^A_{j \rightarrow i})$ for how well an interface adapter converts an argument for a method in the target interface to that for the relevant method in the source interface.  The conversion probability matrix~$t_{ji}$ is defined in terms of these probabilities, where $t_{ji} = P(C^A_{j \rightarrow i})$.

Given method availability vector~$p_i$, method dependency matrix~$a_{ji}$, and conversion probability matrix~$t_{ji}$, we can now define the adaptation operator~$\otimes$.  Instead of just the method dependency matrix being applied to the method availability vector, the conversion probability matrix must also be applied in conjunction with the method dependency matrix:

\begin{dfn}
  Given method dependency matrix $a_{ji}$, conversion probability matrix $t_{ji}$, and method availability vector $p_i$, the \emph{probabilistic adaptation operator} $\otimes$ is defined as:
  \begin{equation}
    \label{eq:prob-operator}
    (a_{ji}, t_{ji}) \otimes p_i = \prod_{a_{ji}} t_{ji} \, p_i
  \end{equation}
\end{dfn}

\begin{dfn}
A tuple $(a_{ji}, t_{ji})$ of a method dependency matrix and a conversion probability matrix is called a \emph{probabilistic adaptation factor}.  The probabilistic adaptation factor for an interface adapter~$A$ is denoted as $\depend{A}$.
\end{dfn}

It should be emphasized that equation~(\ref{eq:prob-operator}) is only rigorously correct given the following three assumptions.  However, the three assumptions make it possible to feasibly compute $P(V_{i,I})$ from values that can be feasibly measured or estimated a~priori in a rigorously sound manner, instead of having to define an ad~hoc computational framework where definitions are vague in their operational meaning.  While it is not hard to see that the assumptions would not hold for most real systems, it is an open question how closely the probabilistic approach based on these assumptions approximates actual losses due to interface adaptation.

\begin{itemize}
\item The probabilities do \emph{not} depend on the specific arguments.
\item The events for separate methods in the source interface are independent.
\item The probability that an interface adapter can successfully convert arguments and results is independent from the probability that a method in the source interface can handle an argument.
\end{itemize}

It should be noted that equation~(\ref{eq:prob-operator}) is incomplete in that it is ambiguous what the result should be when no $a_{ji}$ is true.  If this is the case, it could be that the method in the target interface can always be implemented regardless of availability of methods in the source interface, or it could be that the method cannot be implemented no matter what.

The workaround is simple: a dummy method is defined for each interface, where the method dependency matrixes follow the same rules.  For the conversion probability matrix, setting $t_{j1}$ to zero for all $j$ would yield the expected results, given the usual convention that an empty product has a value of one~\cite{lang2002-emptyproduct}.\footnote{The values for $t_{1i}$ do not matter except for $i=1$, so they can be arbitrarily set to zero.}  We will denote a method availability vector for interface~$I$ in which all methods are available and can handle all arguments by $\mathbf{1}'_I$, where all components have value one except for the component corresponding to the dummy method, which has value zero.

\subsection{Adapter composition}
\label{sec:prob-composition}

We would like to be able to derive a composite probabilistic adaptation factor from the composition of two probabilistic adaptation factors, which would be equivalent to describing the chaining of two interface adapters as if they were a single interface adapter.

Given interfaces $I_1$, $I_2$, and $I_3$, let the corresponding method availability vectors be $p_i$, $q_j$, and $r_k$.  In addition, let there be interface adapters $A_1$ and $A_2$, where $A_1$ converts $I_1$ to $I_2$ and $A_2$ converts $I_2$ to $I_3$, with corresponding probabilistic adaptation factors $(a_{ji}, t_{ji})$ and $(b_{kj}, u_{kj})$, respectively.  We would like to know how to derive the probabilistic adaptation factor $(c_{ki}, v_{ki})$ that would correspond to an interface adapter equivalent to $A_1$ and $A_2$ chained together.

$c_{ki}$ is obviously derived in the same way as specified by the composition operator in section~\ref{sec:preliminaries}.  As for $v_{ki}$, from equation~(\ref{eq:practical-probability}) and our assumptions:

\begin{eqnarray}
  \nonumber  r_k & = &  \prod_{b_{kj}} u_{kj} \, q_j \\
  \nonumber & = & \prod_{b_{kj}} \left( u_{kj} \prod_{a_{ji}} t_{ji} \, p_i \right) \\
  \nonumber & = & \prod_{b_{kj}} \prod_{a_{ji}} u_{kj}  \, t_{ji} \, p_i \\
  & = & \prod_{b_{kj} \wedge a_{ji}} u_{kj}  \, t_{ji} \, p_i \label{eq:prob-wanted-comp-res}
\end{eqnarray}

We want the above to be equivalent to the following:

\begin{eqnarray}
  \nonumber r_k & = & \prod_{c_{ki}} v_{ki} \, p_i \\
  & = & \prod_{\bigvee_j (b_{kj} \wedge a_{ji})} v_{ki} \, p_i \label{eq:prob-wanted-comp-form}
\end{eqnarray}

The composition operator is derived by carefully considering the terms in equations~(\ref{eq:prob-wanted-comp-res}) and~(\ref{eq:prob-wanted-comp-form}), based on collecting the terms for fixed $i$.

If we collect the terms in equation~(\ref{eq:prob-wanted-comp-res}) with fixed $i$, we have (\ref{eq:prob-wanted-terms}).  It should be emphasized that (\ref{eq:prob-wanted-terms}) is \emph{not} identical to (\ref{eq:prob-wanted-comp-res}): the former is a product over varying $j$ with \emph{both} $i$ and $k$ fixed, while the latter is a product over varying $i$ and $j$ with only $k$ fixed.  Also note that if $b_{kj} \wedge a_{ji}$ are all false for varying $j$, then no terms affect the result of (\ref{eq:prob-wanted-comp-res}).  This would be equivalent to (\ref{eq:prob-wanted-terms}) having a value of one, which is expected from an empty product.

\begin{equation}
  \label{eq:prob-wanted-terms}
  \prod_{b_{kj} \wedge a_{ji}} u_{kj}  \, t_{ji} \, p_i
\end{equation}

On the other hand, consider the term in equation~(\ref{eq:prob-wanted-comp-form}) with fixed $i$.  If $\bigvee_j (b_{kj} \wedge a_{ji})$ is false, \ie $b_{kj} \wedge a_{ji}$ are all false for varying $j$, then the term is excluded from the product and is equivalent to multiplying by one, instead.  If it is true, on the other hand, then $v_{ki} \, p_i$ is the term that corresponds to the fixed $i$.  So if we set $v_{ki} \, p_i$ according to (\ref{eq:prob-comp}),\footnote{Remember that only $k$ is fixed in (\ref{eq:prob-wanted-comp-res}) and (\ref{eq:prob-wanted-comp-form}), but both $k$ and $i$ are fixed in (\ref{eq:prob-comp}).} then equations~(\ref{eq:prob-wanted-comp-form}) and~(\ref{eq:prob-wanted-comp-res}) end up having the exact same values.

\begin{equation}
  \label{eq:prob-comp}
  v_{ki} = \prod_{b_{kj} \wedge a_{ji}} u_{kj}  \, t_{ji}
\end{equation}

From this, we can conclude that the composition operator~$\otimes$ for two probabilistic adaptation factors should be defined as in definition~\ref{dfn:prob-comp-operator}:

\begin{dfn}
  \label{dfn:prob-comp-operator}
  Given probabilistic adaptation factors $(b_{kj}, u_{kj})$ and $(a_{ji}, t_{ji})$, the \emph{probabilistic composition operator} $\otimes$ is defined as:
 \begin{equation}
   \label{eq:prob-comp-operator}
   (b_{kj}, u_{kj}) \otimes (a_{ji}, t_{ji}) =
   (b_{kj} \otimes a_{kj}, \prod_{b_{kj} \wedge a_{ji}} u_{kj}  \, t_{ji})
 \end{equation}
\end{dfn}
 
The $\otimes$ operator is ``associative'' when applied to a probabilistic adaptation factors and a method availability vector:\footnote{It is technically not associative in this context since the $\otimes$ operator in $(b_{kj}, u_{kj}) \otimes (a_{ji}, t_{ji}) $ is not the same as the $\otimes$ operator in $(a_{ji}, t_{ji}) \otimes p_i$.}

\begin{thm}
Applying the adaptation operator twice to a method availabiliy vector is the same as applying the composition operator and then applying the adaptation operator:
\[ (b_{kj}, u_{kj}) \otimes ((a_{ji}, t_{ji}) \otimes p_i) = ((b_{kj}, u_{kj}) \otimes (a_{ji}, t_{ji})) \otimes p_i \]
\end{thm}
\begin{proof}
\begin{eqnarray*}
  (b_{kj}, u_{kj}) \otimes ((a_{ji}, t_{ji}) \otimes p_i)
  & = & (b_{kj}, u_{kj}) \otimes \prod_{a_{ji}} t_{ji} \, p_i \\
  & = & \prod_{b_{kj}} u_{kj} \prod_{a_{ji}} t_{ji} \, p_i \\
  & = & \prod_{b_{kj}} \prod_{a_{ji}} u_{kj}  \, t_{ji} \, p_i \\
  & = & \prod_{b_{kj} \wedge a_{ji}} u_{kj}  \, t_{ji} \, p_i \\
  & = & \prod_{\bigvee_j (b_{kj} \wedge a_{ji})} \prod_{b_{kj} \wedge a_{ji}} u_{kj}  \, t_{ji} \, p_i \\
  & = & \prod_{b_{kj} \otimes a_{ji}} \left( \prod_{b_{kj} \wedge a_{ji}} u_{kj}  \, t_{ji} \right) p_i \\
  & = & (b_{kj} \otimes a_{ji}, \prod_{b_{kj} \wedge a_{ji}} u_{kj}  \, t_{ji}) \otimes p_i \\
  & = & ((b_{kj}, u_{kj}) \otimes (a_{ji}, t_{ji})) \otimes p_i
\end{eqnarray*}
\end{proof}

Likewise, probabilistic adaptation factor composition is associative:

\begin{thm}
The composition operator for probabilistic adaptation factors is associative:
\[ (c_{lk}, v_{lk}) \otimes ((b_{kj}, u_{kj}) \otimes (a_{ji}, t_{ji}))
= ((c_{lk}, v_{lk}) \otimes (b_{kj}, u_{kj})) \otimes (a_{ji}, t_{ji}) \]
\end{thm}
\begin{proof}
Using the fact that $b_{kj} \otimes a_{ji} = \bigvee_j (b_{kj} \wedge a_{ji})$ must be true if $b_{kj} \wedge a_{ji}$ is true, we have:
\begin{eqnarray*}
  \lefteqn{(c_{lk}, v_{lk}) \otimes ((b_{kj}, u_{kj}) \otimes (a_{ji}, t_{ji}))} \\
  & = & (c_{lk}, v_{lk}) \otimes (b_{kj} \otimes a_{kj}, \prod_{b_{kj} \wedge a_{ji}} u_{kj}  \, t_{ji}) \\
  & = & (c_{lk} \otimes b_{kj} \otimes a_{kj}, \prod_{c_{lk} \wedge (b_{kj} \otimes a_{kj})} v_{lk} \prod_{b_{kj} \wedge a_{ji}} u_{kj}  \, t_{ji}) \\
  & = & (c_{lk} \otimes b_{kj} \otimes a_{kj}, \prod_{c_{lk} \wedge b_{kj} \wedge a_{ji}  \wedge (b_{kj} \otimes a_{kj})} v_{lk} \, u_{kj}  \, t_{ji}) \\
  & = & (c_{lk} \otimes b_{kj} \otimes a_{kj}, \prod_{c_{lk} \wedge b_{kj} \wedge a_{ji}} v_{lk} \, u_{kj}  \, t_{ji}) \\
  & = & (c_{lk} \otimes b_{kj}  \otimes a_{ji}, \prod_{(c_{lk} \otimes b_{kj}) \wedge c_{lk} \wedge b_{kj} \wedge a_{ji}} v_{lk} \, u_{kj} \, t_{ji} \\
  & = & (c_{lk} \otimes b_{kj}  \otimes a_{ji}, \prod_{(c_{lk} \otimes b_{kj}) \wedge a_{ji}} \left( \prod_{c_{lk} \wedge b_{kj}} v_{lk} \, u_{kj} \right) t_{ji} \\
 & = & (c_{lk} \otimes b_{kj}, \prod_{c_{lk} \wedge b_{kj}} v_{lk} \, u_{kj}) \otimes (a_{ji}, t_{ji}) \\
  & = & ((c_{lk}, v_{lk}) \otimes (b_{kj}, u_{kj})) \otimes (a_{ji}, t_{ji})
\end{eqnarray*}
\end{proof}

However, probabilistic adaptation factor composition is \emph{not} commutative, as can be easily seen by considering the composition of probabilistic adaptation factors whose components are not square matrixes.

We can also show a monotonicity property, which formalizes the notion that extending an interface adapter chain results in worse adaptation loss:

\begin{thm}
If $A_1$ and $A_2$ are interface adapters, where $A_1$ converts $I_1$ to $I_2$ and $A_2$ converts $I_2$ to $I_3$, with $(a_{ji}, t_{ji}) = \depend{A_1}$ and $(b_{kj}, u_{kj}) = \depend{A_2}$ where they follow the rules for the dummy method in sections~\ref{sec:preliminaries}, let $p_k = (b_{kj}, u_{kj}) \otimes \mathbf{1}'_{I_2}$ and $p'_k = (b_{kj}, u_{kj}) \otimes (a_{ji}, t_{ji}) \otimes \mathbf{1}'_{I_1}$.  Then
\[ p'_k \leq p_k \]
\end{thm}
\begin{proof}
From our assumptions, we have:
\[ p_1 = p'_1 = 0 \]

\begin{equation}
  \label{eq:prob-p-k}
  p_k = \prod_{j \neq 1 \wedge b_{kj}} u_{kj}
\end{equation}

\begin{equation}
  \label{eq:prob-p-dash-k}
  p'_k = \prod_{i \neq 1 \wedge b_{kj} \wedge a_{ji}} u_{kj} \, t_{ji}
  = \prod_{b_{kj}} \prod_{i \neq 1 \wedge a_{ji}} u_{kj} \, t_{ji}
  = \prod_{b_{kj}}  \left( u_{kj} \prod_{i \neq 1 \wedge a_{ji}} t_{ji} \right)
\end{equation}

If method~$k$ can never be implemented given the source interface, then $b_{k1}$ will be true, and given that $u_{k1}$ will be zero, $p'_k$ will also have to be zero.  Otherwise, $b_{k1}$ will be false, so we can do a term by term comparison of equations~(\ref{eq:prob-p-k}) and~(\ref{eq:prob-p-dash-k}), taking advantage of the fact that $u_{kj}$ and $t_{ji}$ are probabilities so that they are greater than or equal to zero and lesser than or equal to one:

\[ 0 \leq \prod_{i \neq 1 \wedge a_{ji}} t_{ji} \leq 1 \]
\[ u_{kj} \prod_{i \neq 1 \wedge a_{ji}} t_{ji}  \leq u_{kj} \]

\begin{equation}
  \label{eq:prob-monotonicity}
  \therefore \; p'_k \leq p_k
\end{equation}
\end{proof}

The definitions of the method dependency matrix and the method availability vector in section~\ref{sec:prob-model}, along with the associativity rules proven in this section, provide a succinct way to mathematically express and analyze the chaining of lossy interface adapters using a probabilistic approach.

\subsection{An example}
\label{sec:example}

\begin{figure}
  \centering
  \includegraphics[width=10cm]{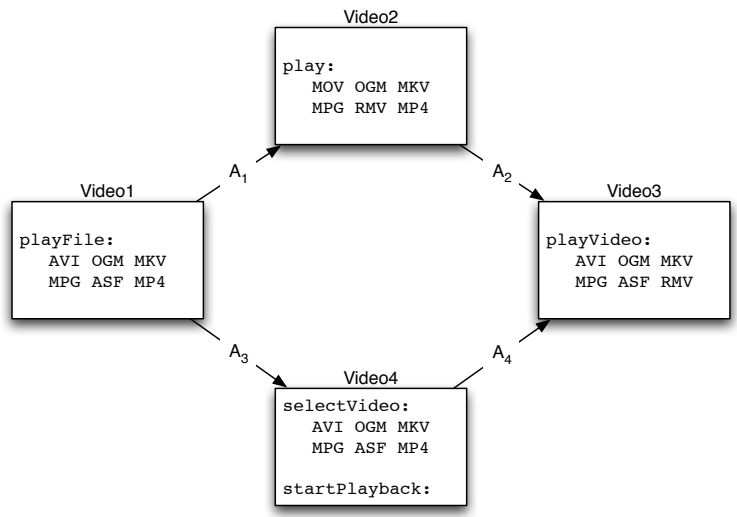}
  \caption{Adapting playback methods in video playback interfaces which support different video formats, expanded version of figure~\ref{fig:example}.}
  \label{fig:expanded-example}
\end{figure}

As an example, we apply the probabilistic approach to analyzing lossy interface chaining to the interface adapter graph of figure~\ref{fig:expanded-example}, which is a slightly expanded version of figure~\ref{fig:example}.  Instead of simply four interfaces each having a single playback method, one of the interfaces, \textit{Video4}, consists of two methods: the \texttt{select\-Video} method chooses a video file that should be played back, and the \texttt{start\-Playback} method actually begins video playback.  As with the example in figure~\ref{fig:example}, each interface can handle different video formats.

In this hypothetical scenario, there is an application written for interface \textit{Video3} which needs to use a video service that actually conforms to \textit{Video1}.  An interface adapter chain from \textit{Video1} to \textit{Video3} would be required if the application is to use the video service.  Since there are two possible interface adapter chains, one which goes through \textit{Video2} and another which goes through \textit{Video4}, we would want to use the chain that can support more video formats.

The interface adapter from \textit{Video1} to \textit{Video2} will be denoted $A_1$, the one from \textit{Video2} to \textit{Video3} will be denoted $A_2$, the one from \textit{Video1} to \textit{Video4} will be denoted $A_3$, and the one from \textit{Video4} to \textit{Video3} will be denoted $A_4$.  The method dependency matrix and conversion probability matrix for adapter~$A_k$ will be denoted $a^k_{ji}$ and $t^k_{ji}$, respectively.  For each interface adapter, we assume that all methods in the target interface can be implemented in terms of all the methods in the source interface.  For simplicity, we will not define dummy methods for any of the interfaces.

Since the single method \texttt{play} of \textit{Video2} depends only on the single method \texttt{play\-File} of \textit{Video1} for $A_1$, $a^1_{ji}$ only has a single true component.  The same is true for $a^2_{ji}$.  On the other hand, the \texttt{select\-Video} and \texttt{start\-Playback} methods of \textit{Video4} both depend on the single method \texttt{play\-File} of \textit{Video1} for $A_3$, so $a^3_{ji}$ has two rows corresponding to the methods in the target interface, each with a single true component corresponding to the method in the source interface.  The \texttt{play\-Video} method of \textit{Video3} depends on both methods of \textit{Video4}, so $a^4_{ji}$ has a single row with two true components.  The method dependency matrixes for each interface adapter are shown below:

\begin{displaymath}
  \begin{array}{cccc}
    a^1_{ji} = \left( \begin{array}{c} t \end{array} \right) &
    a^2_{ji} = \left( \begin{array}{c} t \end{array} \right) &
    a^3_{ji} = \left( \begin{array}{c} t \\ t \end{array} \right) &
    a^4_{ji} = \left( \begin{array}{cc} t & t \end{array} \right)
  \end{array}
\end{displaymath}

As for the conversion probability matrixes, a way to estimate the necessary probabilities is to compare the number of video formats each interface supports.\footnote{While this will not be accurate, it would be a relatively easy way to obtain a rough estimate that could be used for comparing the quality of different interface adapter chains.}  For $A_1$, among the formats \texttt{MOV}, \texttt{OGM}, \texttt{MKV}, \texttt{MPG}, \texttt{RMV}, and \texttt{MP4} that \textit{Video2} should be able to support, the adapted interface can only support \texttt{OGM}, \texttt{MKV}, \texttt{MPG}, \texttt{MP4} since these are supported by the source interface \textit{Video1}, so the conversion probability can be estimated as $\frac{4}{6}$.  Assuming that \texttt{start\-Playback} in \textit{Video4} has no arguments to be converted, the conversion probability matrixes can be set as in the following:

\begin{displaymath}
  \begin{array}{cccc}
    t^1_{ji} = \left( \begin{array}{c} \frac{4}{6} \end{array} \right) &
    t^2_{ji} = \left( \begin{array}{c} \frac{4}{6} \end{array} \right) &
    t^3_{ji} = \left( \begin{array}{c} 1 \\ 1 \end{array} \right) &
    t^4_{ji} = \left( \begin{array}{cc} \frac{5}{6} & 1 \end{array} \right)
  \end{array}
\end{displaymath}

We will first look at the interface adapter chain that starts from \textit{Video1}, passes through \textit{Video2}, and ends at \textit{Video3}.  Given a service conforming to \textit{Video1} that is fully functional, \ie supports all arguments it could receive, the sole component of the method availability vector corresponding to \textit{Video1} is a probability of one.  To see how the interface adapter chain formed from $A_1$ and $A_2$ adapts \textit{Video1} to \textit{Video3}, \ie the result of applying $A_1$ to \textit{Video1} and then applying $A_2$, we can use the adaptation operator:

\begin{displaymath}
  (a^2_{kj}, t^2_{kj}) \otimes (a^1_{kj}, t^1_{ji}) \otimes \left( \begin{array}{c} 1 \end{array} \right) = \left( \begin{array}{c} \frac{4}{9} \end{array} \right)
\end{displaymath}

We can also do the same thing for the interface adapter chain that starts from \textit{Video1}, passes through \textit{Video4}, and ends at \textit{Video3}, \ie the interface adapter chain formed from $A_3$ and $A_4$:

\begin{displaymath}
 (a^4_{kj}, t^4_{kj}) \otimes (a^3_{kj}, t^3_{ji}) \otimes \left( \begin{array}{c} 1 \end{array} \right) = \left( \begin{array}{c} \frac{5}{6} \end{array} \right)
\end{displaymath}

These results roughly estimate that when providing \textit{Video3} by adapting \textit{Video1}, the chain formed from $A_1$ and $A_2$ would allow the interface to handle about $\frac{4}{9}$ of the video files it is asked to play back, while the chain formed from $A_3$ and $A_4$ would allow the interface to handle about $\frac{5}{6}$ of the video files it is asked to play back.  This is consistent with how the former chain is worse in terms of only being able to handle \texttt{OGM}, \texttt{MKV}, and \texttt{MPG}, while the latter chain can handle significantly more formats, specifically \texttt{AVI}, \texttt{OGM}, \texttt{MKV}, \texttt{MPG}, and \texttt{ASF}.\footnote{While the example here is simple enough that we can easily figure out exactly what types of arguments can be handled, it can be prohibitively difficult to do so in the general case~\cite{l:chung2010:absintp}.}  In contrast, the discrete approach would tell us that the two chains are exactly the same.

By using probability estimates of how well each interface adapter can adapt a source interface to a target interface, the probabilistic analysis scheme for interface adapter chaining outlined in this paper can be used to compare the quality of an interface adapter chains where methods may not be adapted perfectly, in contrast to the discrete approach where methods are assumed to be adapted perfectly if they can be adapted at all.

\section{Probabilistic optimal adapter chaining}
\label{sec:complexity-probabilistic}

Like the optimal adapter chaining problem with the discrete chain approach, the optimal adapter chaining problem with the probabilistic approach is NP-complete as well.  This is intuitively the case since the probabilistic approach should be able to encompass the discrete approach, and we show this formally in this section.

We first formally define the optimal adapter chaining problem in the probabilistic approach, which we will call PROB-CHAIN.   Let us have an interface adapter graph $(\{I_i\}, \{A_i\})$, where $\{I_i\}$ is the set of interfaces and $\{A_i\}$ is the set of interface adapters.  Let $f^k$ be the probabilistic adaptation factor associated with adapter~$A_k$.  Let $S \in \{I_i\}$ be the source interface and $T \in \{I_i\}$ be the target interface.  Let $\{ r_m \}$ be the relative invocation probabilities for the methods in the target interface such that $\sum_m r_m = 1$.  Then the problem is whether there is an interface adapter chain $[A_{P(1)}, A_{P(2)}, \ldots, A_{P(n)}]$ such that the source of $A_{P(1)}$ is $S$, the target of $A_{P(n)}$ is $T$, and $\sum_m r_m \, v^T_m$  is at least as large as some probability $X$, where $v^T = f^{P(n)} \otimes \cdots \otimes f^{P(2)} \otimes f^{P(1)} \otimes \mathbf{1}'_S$.

Informally, this is an optimization problem which tries to maximize the probability that an argument can be handled by a method in a fixed target interface, obtained by applying an interface adapter chain on a fully-functional service which conforms to the source interface.  $\{ r_m \}$ would express how often methods are invoked relative to each other.

\begin{thm}
There is a reduction from the discrete approach to the probabilistic approach for analyzing loss in interface adapter chains.
\end{thm}
\begin{proof}
Let there be a method availability vector~$p_i$ and a method dependency matrix~$a_{ji}$ as expressed in the discrete approach.  We construct corresponding method availability vector~$p'_i$, method dependency matrix~$a'_{ji}$, and conversion probability matrix~$t'_{ji}$ as expressed in the probabilistic approach as follows.  If $p_i$ is true, then set $p'_i$ to one, else set $p'_i$ to zero.  $a'_{ji}$ is just the same as $a_{ji}$.  And set all $t'_{ji}$ to one.  Then we have:
\[ a_{ji} \otimes p_i = \bigwedge_j (a_{ji} \rightarrow p_i) = \bigwedge_{a_{ji}} p_i \]
\[ (a'_{ji}, t'_{ji}) \otimes p'_i = \prod_{a'_{ji}} t'_{ji} \, p'_i = \prod_{a_{ji}} p'_i \]
and it  is easy to see that a component of $a_{ji} \otimes p_i$ is true if and only if the corresponding component of $(a'_{ji}, t'_{ji}) \otimes p'_i$ is one, and that a component of $a_{ji} \otimes p_i$ is false if and only if the corresponding component of $(a'_{ji}, t'_{ji}) \otimes p'_i$ is zero.

This shows how an interface adapter graph for the discrete approach can be converted to one for the probabilistic approach in a way that the adaptation operators in both approaches basically have the same behavior.  Since all the mathematics for both approaches follow from the definition of the adaptation operators, we have just shown that the probabilistic approach can encompass the discrete approach.
\end{proof}

Next, we formally describe the equivalent problem for the discrete approach, which we will call CHAIN and is NP-complete~\cite{chung:ietsoftw2010}.  Let us have an interface adapter graph $(\{I_i\}, \{A_i\})$, where $\{I_i\}$ is the set of interfaces and $\{A_i\}$ is the set of interface adapters.  Let $a^k$ be the method dependency matrix associated with adapter~$A_k$.  Let $S \in \{I_i\}$ be the source interface and $T \in \{I_i\}$ be the target interface.  Then the problem is whether there is an interface adapter chain $[A_{P(1)}, A_{P(2)}, \ldots, A_{P(m)}]$ such that the source of $A_{P(1)}$ is $S$, the target of $A_{P(m)}$ is $T$, and $\|v^T\| = \|a^{P(m)} \otimes \cdots \otimes a^{P(2)} \otimes a^{P(1)} \otimes \mathbf{1}'_S\|$ is at least as large as some parameter $N$.

\begin{thm}
PROB-CHAIN is NP-complete.
\end{thm}
\begin{proof}
Given $M$ methods in the target interface, use the method described above to convert an input for CHAIN to an input for PROB-CHAIN, where we also set all $r_m$ to $\frac{1}{M}$.  Then $\sum_m r_m \, v^T_m$ will be $\frac{n}{M}$, where $n$ is the number of methods available from the interface adapter chain, so PROB-CHAIN with $X$ set to $\frac{N}{M}$ will solve CHAIN.  Since CHAIN is NP-complete and it is easy to verify if an alternate chain results in smaller $\sum_m r_m \, v^T_m$, PROB-CHAIN must also be NP-complete.
\end{proof}

\section{A greedy algorithm}
\label{sec:greedy-probabilistic-algorithm}

As shown in section~\ref{sec:complexity-probabilistic}, the problem of finding an optimal interface adapter chain maximizing the probability of an argument being handled by a method in the target interface is an NP-complete problem.  Short of developing a polynomial-time algorithm for an NP-complete problem, practical systems will have to use a heuristic algorithm or an exponential-time algorithm with reasonable performance in practice.

Algorithm~\ref{alg:prob-greedy-algorithm} is a greedy algorithm that finds an optimal interface adapter chain between a given source interface and a target interface.  Given an interface adapter graph~$G$, it works by looking at every possible acyclic adapter chain with an arbitrary source that results in the target interface~$t$ in order of increasing loss, taking advantage of equation~(\ref{eq:prob-monotonicity}), until we find a chain that starts with the desired source interface~$s$.

In this context, loss means the probability that a method in the target interface \emph{cannot} handle an argument given a fully functional service with the source interface, which is computed in algorithm~\ref{alg:prob-loss}, so the algorithm is guaranteed to find the optimal interface adapter chain.  In the worst case, however, the algorithm takes exponential time since there can be an exponential number of acyclic chains in an interface adapter graph.

\begin{algorithm}
  \caption{A probabilistic greedy algorithm for interface adapter chaining.}
  \label{alg:prob-greedy-algorithm}
  \begin{algorithmic}
  \Procedure{Prob-Greedy-Chain}{$G = (V, E)$, $s$, $t$, $\{ r_m \}$}
    \State $\mathit{C} \gets \{ [] \}$
    \Comment{chains to extend}
    \State $\mathit{M} = \emptyset$
    \Comment{discarded chains}
    \State $D \gets \{[] \mapsto \mathbf{I}_{\dim(\mathbf{1}'_t)} \}$
    \Comment{method dependency matrixes}
    \While{$C \neq \emptyset$}
      \State $c \gets \mbox{element of $C$ minimizing $\textsc{Prob-Loss}(c,D,\{r_m\})$}$
      \If{$c \neq [] \wedge \source{c[1]} = s$}
        \State \textbf{return} $c$
      \ElsIf{no acyclic chain not in $C \cup M$ extends $c$}
        \State $C \gets C - \{c\}$
        \State $M \gets M \cup \{c\}$
      \Else
        \If{$c = []$}
          \State $B \gets \{ [e] \,|\, e \in E, \target{e} = t\}$
        \Else
          \State $B \gets \{ e : c \,|\, e \in E, \target{e} = \source{c[1]} \}$
        \EndIf
        \State remove cyclic chains from $B$
        \State $C \gets C \cup B$
        \State $D \gets D \cup \{ e:c \mapsto D[c] \otimes \depend{e} \,|\, e : c \in B \}$
      \EndIf
    \EndWhile
  \EndProcedure
  \end{algorithmic}
\end{algorithm}

\begin{algorithm}
  \caption{Computing the probabilistic loss of an interface adapter chain.}
  \label{alg:prob-loss}
  \begin{algorithmic}
  \Function{Prob-Loss}{$c$, $D$, $\{ r_m \}$}
    \State $s \gets \source{c[1]}$
    \State $v \gets D[c] \otimes \mathbf{1}'_s$
    \State \textbf{return} $1 -  \sum_m r_m \, v_m$
  \EndFunction
  \end{algorithmic}
\end{algorithm}

Algorithm~\ref{alg:prob-greedy-algorithm} can be easily extended to support behavior similar to service discovery by checking whether the current source is among a potential set of source interfaces instead of just checking against one, as is done with a similar algorithm based on the discrete approach~\cite{chung:ietsoftw2010}.

\section{Conclusions}
\label{sec:conclusions}

Interface adapters can allow code written to use one interface to use another interface, and chaining them together can substantially reduce the effort required to create interface adapters.  Since interface adapters will often be unable to convert interfaces perfectly, loss can be incurred during interface adaptation, and we need a rigorous mathematical framework for analyzing such loss.  Instead of just analyzing whether or not a method in a target interface can be provided, we have developed a probabilistic framework where partial adaptation of methods can also be handled.

We developed the probabilistic framework by first constructing a probabilistic model for interface adaptation.  Based on this, we defined mathematical objects and operations which probabilistically express loss in adapted interfaces and interface adapters, which were then used to prove that probabilistic optimal adapter chaining is NP-complete and to construct a greedy algorithm which can construct an optimal adapter chain with exponential time in the worst case.  These provide a more fine-grained approach to analyzing loss in interface adapter chains compared to a discrete approach.

Future avenues of research include alternate probabilistic approaches which require weaker and more realistic assumptions that can still be feasibly used in real interface adaptation systems.  Another avenue of research is to find good ways to derive the necessary probabilities from the interface adapters, either through empirical means where interface adapters are invoked on many arguments to measure the probabilities or through analytical means which can approximate the probabilities based on program structure.  Finally, there remains the design and implementation of an actual interface adaptation system which takes advantage of the probabilistic approach to analyzing loss in interface adapter chaining.

\bibliographystyle{plain}
\bibliography{strings,articles,hearsay,local,proceedings,books}

\end{document}